\documentclass[letterpaper]{article} 
\usepackage{aaai2026}  
\usepackage{times}  
\usepackage{helvet}  
\usepackage{courier}  
\usepackage[hyphens]{url}  
\usepackage{graphicx} 
\usepackage{natbib}  
\usepackage{caption} 
\usepackage{algorithm}
\usepackage{algorithmic}

\usepackage{multirow}
\usepackage{amssymb}
\usepackage{amsmath}
\usepackage{tabularx}
\usepackage{array}
\usepackage{amsthm}
\usepackage{subcaption}

\usepackage{newfloat}
\usepackage{listings}
\DeclareCaptionStyle{ruled}{labelfont=normalfont,labelsep=colon,strut=off} 
\floatstyle{ruled}
\newfloat{listing}{tb}{lst}{}
\floatname{listing}{Listing}
\title{Adaptive Diffusion-based Augmentation for Recommendation}
\author{
    Na Li\textsuperscript{\rm 1},
    Fanghui Sun\textsuperscript{\rm 1},
    Yan Zou\textsuperscript{\rm 1},
    Yangfu Zhu\textsuperscript{\rm 2},
    Xiatian Zhu\textsuperscript{\rm 3},
    Ying Ma\textsuperscript{\rm 1}\thanks{Corresponding author.}\\
}
\affiliations{
    \textsuperscript{\rm 1}Harbin Institute of Technology, Heilongjiang, China\\
    \textsuperscript{\rm 2}Capital Normal University, Beijing, China\\
    \textsuperscript{\rm 3}University of Surrey, Guildford, United Kingdom\\
    \{24s003043, 23b903066\}@stu.hit.edu.cn, \{Sunfanghui, y.ma\}@hit.edu.cn,\\
    zhuyangfu@cnu.edu.cn, eddy.zhuxt@gmail.com
}

\begin{document}

\maketitle

\begin{abstract}
Recommendation systems often rely on implicit feedback, where only positive user-item interactions can be observed.
Negative sampling is therefore crucial to provide proper negative training signals.
However, existing methods tend to mislabel potentially positive but unobserved items as negatives and lack precise control over negative sample selection. 
We aim to address these by generating controllable negative samples, rather than sampling from the existing item pool.
In this context, we propose \textbf{A}daptive \textbf{D}iffusion-based \textbf{A}ugmentation for \textbf{R}ecommendation (\textbf{ADAR}), a novel and model-agnostic module that leverages diffusion to synthesize informative negatives.
Inspired by the progressive corruption process in diffusion, ADAR simulates a continuous transition from positive to negative, allowing for fine-grained control over sample hardness. 
To mine suitable negative samples, we theoretically identify the transition point at which a positive sample turns negative and derive a score-aware function to adaptively determine the optimal sampling timestep.
By identifying this transition point, ADAR generates challenging negative samples that effectively refine the model's decision boundary. 
Experiments confirm that ADAR is broadly compatible and boosts the performance of existing recommendation models substantially, including collaborative filtering and sequential recommendation, without architectural modifications.
\end{abstract}


\section{Introduction}
In modern recommendation systems, implicit feedback \cite{ding2020implicit} has become the predominant signal for learning user preferences, due to its ubiquitous availability in real-world applications. Such feedback typically includes user behaviors like clicks, views, or purchases, which are treated as positive samples. However, explicit negative signals are usually absent, making it challenging to identify items that users truly dislike. Yet the learning process critically depends on the ability to distinguish preferred items from non-preferred ones. In this context, the construction of high-quality negative samples is essential for guiding recommendation models \cite{rendle2014improving}.

To construct negative samples from implicit feedback, most existing models adopt sampling methods that pick negative samples from unobserved user-item interactions. A common method is to perform static random sampling over these unobserved interactions due to its simplicity and efficiency \cite{rendle2012bpr}.
However, this approach cannot provide enough negative signal, limiting model’s ability to learn accurate user preferences. To address this, recent studies in collaborative filtering have explored hard negative sampling strategies \cite{zhang2013dns}, which leverage unobserved items that are incorrectly ranked highly, thereby providing more challenging supervision signals. 
Building on this, more sophisticated sampling strategies such as MixGCF \cite{huang2021mixgcf} and AHNS \cite{lai2024ahns} have been proposed to find hard negative samples by mixing or scoring.
In sequential recommendation, sampling is further challenged by temporal and contextual dynamics of user behavior. GNNO \cite{fan2023gnno} enhances sampling by identifying hard ones based on global transitions, while recent context-aware method \cite{seol2025context} adapts sampling to varying distributions across time.

However, these strategies are suboptimal and suffer from several limitations. First, they tend to introduce false-negative signals, as some unobserved items may actually align with user interests but remain unseen due to limited exposure. Incorrectly treating potentially relevant items as negatives may introduces misleading signals during training, thereby hindering the learning of user preferences.
Second, existing methods lack fine-grained control in the negative sampling process. Although dynamic strategies \cite{zhang2013dns, zhao2023ans} improve upon static heuristics by ranking candidates based on their scores, they lack flexibility in obtaining negative samples with controllable hardness. 
Specifically, the current approaches are unable to precisely regulate the true gap between positive and negative samples, resulting in suboptimal ranking performance.

Instead of relying on negative samples drawn from the existing item pool, we aim to model a continuous transition from positive to negative samples to address these limitations.
To achieve this, we adopt a generative approach based on diffusion processes, which naturally simulate the gradual transformation.
Diffusion models \cite{ho2020ddpm,dhariwal2021diffbeat}, originally designed for high-fidelity image synthesis \cite{zhou2025attention}, gradually transform structured input into noise and learn to reverse this process. 
We observe that, during the forward noising process, a positive sample progressively loses its informative characteristics and eventually becomes indistinguishable from a negative sample. The continuous degradation offers a natural lens for identifying the transition point at which a sample crosses from positive to negative. 

In this paper, we propose \textbf{A}daptive \textbf{D}iffusion-based \textbf{A}ugmentation for \textbf{R}ecommendation (\textbf{ADAR}), a novel and model-agnostic augmentation module that employs diffusion to adaptively synthesize informative negative examples.
Our core idea is to model the degradation of positive samples through a diffusion process. 
During this process, positive samples are progressively corrupted, and a critical transition occurs when their identity shifts from positive to negative.
So we can treat the outputs at each timestep as a candidate pool of negative samples.
Then we provide a theoretical characterization of the transition and introduce a score-aware function to adaptively determine the optimal sampling timestep.
Leveraging this well-defined transition point, ADAR is able to adaptively generate informative and challenging negative samples.
These generated samples can then be used as a plug-and-play augmentation to enhance existing various recommendation models.

Our main contributions are summarized as follows:
\begin{itemize}
    \item We propose ADAR, a novel and model-agnostic augmentation module that utilizes diffusion to generate high-quality negative samples, addressing the limits of false-negative signals and controllability in existing methods.
    \item We theoretically define a transition point within the diffusion process that marks when a positive sample becomes negative, thereby enabling the adaptively choose of informative negative samples.
    \item ADAR is compatible with a wide range of models, including collaborative filtering and sequential recommendation, demonstrating strong generalizability without any model-specific modifications.
\end{itemize}

\section{Related Work}
\subsection{Negative Sampling}
In recommendation, negative sampling plays a crucial role in constructing training pairs that guide the model to distinguish between preferred and non-preferred items. 
The effectiveness of learning heavily depends on the quality and hardness of these negative samples. 
Simple uniform sampling over unobserved interactions \cite{rendle2012bpr} often yields trivial negatives and weak supervision. 
For better discriminative user/item representations, hard negative sampling selects challenging negatives ranked highly by the model \cite{zhang2013dns}.
MixGCF \cite{huang2021mixgcf} and DropMix \cite{ma2023dropmix} obtain negative samples through a dimensional mixing mechanism. Going a step further, DINS \cite{wu2023dins} and TriSampler \cite{yang2024trisampler} expand the sampling area to spatial and triangular areas. AHNS \cite{lai2024ahns} can select existing negative samples of different hardness.

In sequential recommendation, negative sampling must handle temporal and contextual dynamics. GNNO \cite{fan2023gnno} mines hard negatives based on global item transitions, while context-aware methods \cite{seol2025context} adapt sampling to context-dependent item distributions, improving recommendation quality under dynamic user behaviors.

However, these methods lack a principled way to model the continuous transition between positive and negative samples, leading to suboptimal model performance.

\subsection{Diffusion Model}
Diffusion model \cite{ho2020ddpm,dhariwal2021condiff2,ho2022condiff1} is a generative modeling method that progressively adds noise to data and learns to reverse this process, which can effectively capture data distribution for tasks such as images \cite{yi2024diffimage}, text \cite{arriola2025diffblock}, or molecular generation \cite{wang2024diffpep}.
In recommendation domain, diffusion models are commonly used to model user preferences.
For instance, DiffuRec \cite{li2023diffurec} employs denoising diffusion to model sequential recommendation with item information. 
DiQDiff \cite{mao2025diqdiff} guides diffusion model through semantic codebook. 
CCDRec \cite{yang2025ccdrec} proposes a conditioned diffusion for multi-model recommendation with curriculum learning. 
Moreover, potential of diffusion models for negative sampling has gained attention across different domains. 
DMNS \cite{nguyen2024dmns} formulates multi-level sampling via conditional diffusion for link prediction,
and MMKGC \cite{niu2025mmkgc} extends this to knowledge graph by incorporating multi-model semantics.

However, existing diffusion-based samplers typically rely on fixed timesteps and are often tailored to specific domains.
In contrast, our method introduces a transition-point-driven sampling, which can adaptively select negative samples and serve as a plug-and-play module to enhance models.

\begin{figure*}[!t] 
    \centering
    \includegraphics[width=0.9\linewidth]{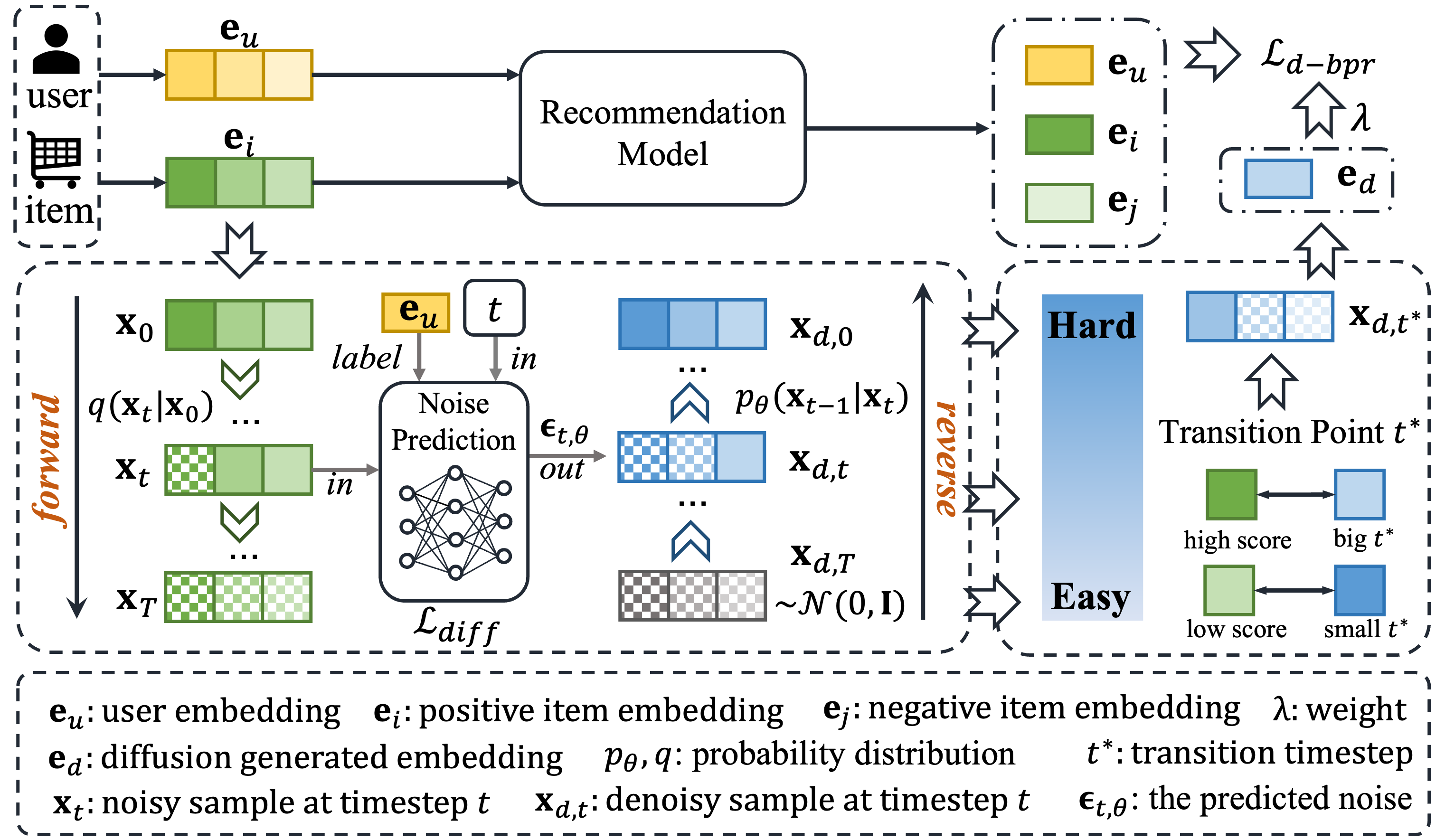}
    \caption{Overview of the proposed ADAR: ADAR leverages the generative dynamics of diffusion models to produce meaningful negatives with controllable hardness. First, we apply a diffusion model to generate negative candidates by gradually corrupting positive samples. Then, we determine the transition point $t^*$ through a score-aware function to adaptively choose the augmented negative. ADAR can serve as a plug-and-play module to enhance diverse recommendation models.}
    \label{fig:main_figure}
\end{figure*}

\section{Preliminary}
\subsection{Problem Formulation}
We denote the set of historical user-item interactions by $\mathcal{D} ^+=\{(u,i)|u\in \mathcal{U},i\in \mathcal{I}\}$, where $\mathcal{U}$ and $\mathcal{I}$ are the set of users and the set of items, respectively. 
We consider all unobserved interactions $\mathcal{D}^- =\{(u,j)|u\in \mathcal{U},j\notin \mathcal{D}_u ^+\}$ as candidates for negative samples.
Standard recommendation models optimize a pairwise loss by training on triplets $(u,i,j)$, where $i$ is a positive item and $j$ is a negative item sampled from $\mathcal{D}^-$. 
Our goal is to enhance model performance by introducing extra negative samples generated through diffusion-based perturbation of positive sample $i$. 
These samples can simulate hard negatives, thus encouraging model to learn more discriminative representations.

\subsection{Diffusion Model}
Diffusion model \cite{ho2020ddpm,dhariwal2021diffbeat} runs a forward process to add Gaussian noise to the input, and then performs a reverse process to denoise and reconstruct the original input in an iterative manner. Given an input $\mathbf{x}_0 \sim q(\mathbf{x}_0)$ and total number of diffusion steps $T$ to get subsequent state $\mathbf{x}_1,\mathbf{x}_2,...,\mathbf{x}_T$, the forward process of the diffusion model can be defined as
\begin{equation}
    q(\mathbf{x}_t|\mathbf{x}_0)=\mathbf{\mathcal{N}}(\mathbf{x}_t;\sqrt{\bar{\alpha }_t }\mathbf{x}_0,(1-\bar{\alpha }_t)\mathbf{I}),
\end{equation}
\begin{equation}
    \mathbf{x}_t=\sqrt{\bar{\alpha }_t }\mathbf{x}_0+\sqrt{(1-\bar{\alpha }_t)}\mathbf{\epsilon}_t,
    \label{eq:forward}
\end{equation}
where cumulative signal retention $\bar{\alpha }_t={\textstyle \prod_{s=1}^{t}\alpha_s}$, $\alpha_t=1-\beta_t$, and $\beta_1,\beta_2,...,\beta_T$ controls the amount of added noise. $t \in [0,T]$ represents the index of timesteps. $\epsilon_t$ represents the noise sampled form $\mathbf{\mathcal{N}}(0,\mathbf{I})$. Typically, $\beta_t$ increases monotonically over time according to a predefined schedule (e.g., linear, cosine or sigmoid), resulting in a gradual increase in noise level. 

And the reverse process uses a neural network with parameter $\theta$ to predict the noise added in the forward process. The reverse transformations can be formulated as:
\begin{equation}
    p_\theta(\mathbf{x}_{t-1}|\mathbf{x}_t) = \mathcal{N}(\mathbf{x}_{t-1};\mu_\theta (\mathbf{x}_t,t), {\textstyle \sum_{\theta }(\mathbf{x}_t,t)}).
\end{equation}
In order to model $p_\theta(\mathbf{x}_{t-1}|\mathbf{x}_t)$, DDPM \cite{ho2020ddpm} usually fixes ${\textstyle \sum_{\theta }(\mathbf{x}_t,t)}$ in advance, and use a reparameterization trick to derive $\mu_\theta (\mathbf{x}_t,t)$:
\begin{equation}
    \mu_\theta(\mathbf{x}_t,t)=\frac{1}{\sqrt{\alpha_t}}\mathbf{x}_t-\frac{1-\alpha_t}{\sqrt{\alpha_t} \sqrt{1-\bar{\alpha_t}}}\epsilon_\theta (\mathbf{x}_t,t),
\end{equation}
where $\epsilon_\theta (\mathbf{x}_t,t)$ is the predicted target, which should be similar to Gaussian noise. In actual operation, we start from $\mathbf{x}_T\sim \mathcal{N}(0,\mathbf{I}) $, and then restore $\mathbf{x}_0$ step by step.

\section{Method}
We propose Adaptive Diffusion-based Augmentation for Recommendation (ADAR), a model-agnostic module for generating high-quality negative samples. 
ADAR simulates a gradual degradation of positive samples through a diffusion process and we detail the formulation of this process. 
ADAR's main architecture is shown in Figure~\ref{fig:main_figure}.

\subsection{Diffusion-based Sample Generation}
In this section, we describe how ADAR generates candidate negative samples.
ADAR learns the progressive noising and denoising process, enabling controllable negative sample generation for recommendation.

Let $i$ in $(u,i) \in \mathcal{D}^+$ be a positive sample. The corresponding embedding $\mathbf{e}_i\in \mathbb{R}^d $ is first obtained from an encoder $\phi (i)$, which may be derived from any different model according to the specific recommendation tasks. The user representation $\mathbf{e}_u$ is similarly obtained from $\phi (u)$. 

\textbf{Forward Process.} We apply the forward diffusion process according to Eq.\ref{eq:forward}, producing a sequence of progressively noised embeddings. Specifically, we set $\mathbf{x}_0$ equal to $\mathbf{e}_i$, perform the forward diffusion process, and obtain noisy sample set $\mathcal{X}=\{\mathbf{x}_t\}_{t=0}^T$. This process gradually removes positive information from the original interaction, pushing $\mathbf{x}_t$ toward a standard Gaussian distribution.

\textbf{Noise Prediction Model.}
To generate candidate negative samples, we need to approximate the noise added at each time step, which can be formulated as
\begin{equation}
    \epsilon_{t,\theta} = \tau(\mathbf{x}_t,\mathbf{e}_t,\mathbf{e}_u;\theta),
\end{equation}
where $\mathbf{e}_t=\mathrm{PE}(t)$ is the time embedding, and $\tau$ is the noise prediction model with learnable parameter $\theta$.
Meanwhile, in order to make full use of the user embedding $\mathbf{e}_u$, we use the conditional diffusion model \cite{dhariwal2021condiff2,ho2022condiff1} to input the user embedding as label information to guide the generation.

We adopt sine and cosine coding to construct a continuous, periodic time embedding, following the approach in \cite{ho2020ddpm}. The timestep embedding is formulated as:
\begin{equation}
    \mathrm{PE}(t) = 
    \begin{cases}
        sin(t\cdot{10000^{-\frac{2i}{d_t}}}),\text{if $i$ is even} \\
        cos(t\cdot{10000^{-\frac{2i}{d_t}}}),\text{if $i$ is old} 
    \end{cases},
\end{equation}
where $i\in[0,d-1]$ is the dimension index. This encoding method facilitates the model’s learning of relative timestep relationships.

To better adapt to recommendation tasks, we implement $\tau$ using feature-wise linear modulation (FiLM) \cite{perez2018film}, which dynamically modulates intermediate feature representations based on timestep and label embeddings. This approach enables efficient noise prediction while ensuring a lightweight adaptation to structured recommendation data.
Specifically, $\tau$ is formulated as:
\begin{equation}
    \epsilon_{t,\theta}=\gamma(\mathbf{e}_t,\mathbf{e}_u;\theta_\gamma) \odot \mathbf{x}_{t} + \eta(\mathbf{e}_t,\mathbf{e}_u;\theta_\eta),
\end{equation}
where $\gamma$ and $\eta$ are interpreted as scale and offset, respectively, and $\odot$ represents element-wise multiplication. 
In our model, these components are implemented as a simple multilayer perceptron with multiple fully connected layers.

\textbf{Diffusion Loss.}
Following the standard formulation of diffusion models \cite{dhariwal2021diffbeat}, we employ a mean squared error (MSE) loss to minimize the discrepancy between the predicted noise in the reverse process and the Gaussian noise injected during forward diffusion.
However, in recommendation settings, recovering the noise alone does not ensure that the final generated samples retain the semantic essence of the original positive instances \cite{li2023diffurec}. To address this, we incorporate an additional loss that enforces consistency between the final generated output and its corresponding positive sample.
Specifically, for a given timestep $t$, the diffusion loss function is defined as:
\begin{equation}
    \label{diff_loss}
    \mathcal{L}_{diff}=\left \| \epsilon -\epsilon_{t,\theta} \right \| ^2 + ||\mathbf{x}_0-\mathbf{x}_{d,0}||^2.
\end{equation}

\textbf{Reverse Process.}
After the noise prediction model $\tau$ is trained, we employ it to generate candidate samples via the reverse diffusion process, where each timestep output serves as a potential negative sample. 
Specifically, we start from a Gaussian noise sampled from the standard normal distribution, representing a maximally corrupted version in the latent space:
\begin{equation}
    \mathbf{x}_{d,T} \sim \mathcal{N}(0, \mathbf{I}),
\end{equation}
and progressively apply the reverse update rule:
\begin{equation}
    \mathbf{x}_{d,t-1} = \frac{1}{\sqrt{\alpha_t}} \mathbf{x}_{d,t} - \frac{1 - \alpha_t}{\sqrt{\alpha_t} \sqrt{1 - \bar{\alpha}_t}} \epsilon_{t,\theta},
\end{equation}
for $t = T, T-1, \dots, 1$. 
Each denoising step incrementally restores information, resulting in the denoised sample set $\mathcal{X}_d = \{\mathbf{x}_{d,t}\}_{t=0}^{T}$.
These outputs are novel, plausible candidates resembling corrupted versions of positive items.

\subsection{Transition Point Detection}
However, not all generated samples in $\mathcal{X}_d$ are equally informative for training. Samples that are too similar to the original positives may reinforce redundancy, while those that are too noisy may introduce harmful signals.
To address this, we propose an adaptive mechanism to identify the transition point, the moment when a generated sample shifts from being positive to negative in terms of user preference. 
We select the generated sample at the transition point as the augmented negative sample, which lies near the decision boundary, and thus offer maximum utility for training.

\newtheorem{theorem}{Theorem}
\begin{theorem}
    Let $\mu^+$ be the positive item score. Then the cumulative signal retention $\bar{\alpha}_{t^*}$ value at the transition point $t^*$ is negatively correlated with $\mu^+$.
    \label{thm:negative}
\end{theorem}

\begin{proof}
We begin by modeling the user preference score as a function $f(\mathbf{u}, \mathbf{x})$, where $\mathbf{u}$ and $\mathbf{x}$ denote the user and item representations, respectively (with $\mathbf{u} = \mathbf{e}_u$). As we agreed before, $\mathbf{x}_0$ denote a positive sample, and $\mathbf{x}_t$ its noised version at diffusion step $t$.
Under the assumption that the scoring function $f(\mathbf{u}, \mathbf{x})$ is approximately linear in $\mathbf{x}$ (e.g., inner product models), the expected preference score of the noised item $\mathbf{x}_t$ becomes:
\begin{align}
\mathbb{E}[f(\mathbf{u}, \mathbf{x}_t)] 
&= \mathbb{E}[f(\mathbf{u}, \sqrt{\bar{\alpha}_t} \mathbf{x}_0 + \sqrt{(1 - \bar{\alpha}_t)} \epsilon_t)] \notag \\
&= \sqrt{\bar{\alpha}_t} \cdot f(\mathbf{u}, \mathbf{x}_0) \notag \\
&= \sqrt{\bar{\alpha}_t} \cdot \mu^+ \notag,
\end{align}
where $\mu^+ = f(\mathbf{u}, \mathbf{x}_0)$ is the score of the original positive item, which is assumed to be greater than zero to reflect user preference. 
We define the transition point $t^*$ as the first diffusion step at which the expected user score for $\mathbf{x}_t$ falls below a predefined negative threshold $\mu^-$. That is,
\[
\mathbb{E}[f(\mathbf{u}, \mathbf{x}_{t^*})] = \mu^- \quad \Rightarrow \bar{\alpha}_{t^*} = \left( \frac{\mu^-}{\mu^+} \right)^2.
\]
Taking derivative w.r.t. $\mu^+$:
\[
\frac{d\bar{\alpha}_{t^*}}{d\mu^+} = -2 \left( \frac{\mu^-}{\mu^+} \right)^2 \cdot \frac{1}{\mu^+} < 0,
\]
demonstrating that $\bar{\alpha}_{t^*}$ is strictly decreasing in $\mu^+$.
\end{proof}

Since $\bar{\alpha}_t$ value is a strictly decreasing function of $t$, it follows that $t^*$ must increase with $\mu^+$. That is, positive samples with higher scores retain their information identity longer during diffusion, and undergo transition at later steps.

While this formulation provides theoretical guidance, practical datasets lack explicit negative interactions, making $\mu^-$ unobservable. To circumvent this, we adopt a parameterized heuristic based on the observed positive score $p_s = f(\mathbf{u}, \mathbf{x}_0)$ to approximate $t^*$. Specifically, we define a score-aware function:
\begin{equation}
    t^* = \text{sigmoid}(\omega \cdot \exp(k \cdot p_s)) \cdot T,
\end{equation}
where $w>0$ and $k>0$ are tunable hyperparameters that control the shape and scale of the transition function. 
In practice, we can set $\omega=1$ and $k=1$, which consistently lead to performance improvements across most datasets, as verified by extensive experiments.

This score-aware function aligns with our theoretical result, capturing the negative correlation between the positive score and the cumulative signal retention $\bar{\alpha}_{t^*}$ value. 
By adaptively mapping higher $p_s$ to deeper diffusion step $t^*$, it dynamically adjusts the distance of generated samples according to the observed positive score.

After determining the transition point $t^*$, we then select the sample $\mathbf{x}_{d,t^*}$ at timestep $t^*$ from the denoised sample set $\mathcal{X}_d$ as the final augmented negative.

\subsection{Optimization Objectives}
To optimize the recommendation model, we adopt BPR \cite{rendle2012bpr}, a widely established objective for pairwise learning from implicit feedback. 
Given that ADAR operates as a model-agnostic augmentation module, it can be seamlessly integrated into existing recommendation models without modifying their core structure. 

In our formulation, let $\mathbf{e}_j$ denote the item embedding selected via the baseline negative sampling method, and $\mathbf{e}_d$ denote the embedding of the sample generated through ADAR.
To fully exploit the information of generated negatives, we integrate them into the learning objectives as a supplement. 
Specifically, the total loss is given by
\begin{equation}
    \label{bpr_loss}
    \mathcal{L}_{d-bpr}=-\sum_{u,i,j,d}\ln{\sigma(\mathbf{e}_u^\top \mathbf{e}_{i}-(\mathbf{e}_u^\top \mathbf{e}_{j} + \lambda \mathbf{e}_u^\top \mathbf{e}_d))} ,
\end{equation}
where the weighting factor $\lambda\in[0,1]$ is a hyperparameter that regulates the contribution of the generated samples.
ADAR introduces negative samples with controlled hardness, which encourages the model to construct more discriminative representations. 
We adopt an alternating training strategy between the diffusion model and the encoder.
For a comprehensive understanding of the training procedure of ADAR, the training pseudo-code is provided in Algorithm~\ref{algo:DAR}.

\begin{algorithm}[h]
\caption{ADAR}
\label{algo:DAR}
\begin{algorithmic}[1]
    \STATE \textbf{Input:} Set of Implicit Feedback $\mathcal{D} ^+=\{(u,i)|u\in \mathcal{U},i\in \mathcal{I}\}$, predefined hyperparameters $T$ and $\lambda$.
    \STATE \textbf{Output:} Encoder model, Diffusion model.
    \STATE Initialization
    \WHILE{not converged}
        \FOR{each mini-batch $\mathcal{B}$ sampled from $\mathcal{D} ^+$}
            \STATE Get embeddings of user $u$ and positive item $i$
            \STATE // Diffusion-based sample generation
            \STATE $\mathbf{x}_t=\sqrt{\bar{\alpha }_t }\mathbf{x}_0+\sqrt{(1-\bar{\alpha }_t)}\epsilon_t$, $\epsilon_t\sim\mathcal{N}(0,\mathbf{I})$
            \STATE $\mathbf{e}_t=\mathrm{PE}(t)$
            \STATE Optimization Diffusion by minimizing $\mathcal{L}_{diff}$
            \STATE Iteratively compute the reverse diffusion step following $ \mathbf{x}_{d,t-1} = \frac{1}{\sqrt{\alpha_t}} \mathbf{x}_{d,t} - \frac{1 - \alpha_t}{\sqrt{\alpha_t} \sqrt{1 - \bar{\alpha}_t}} \epsilon_{t,\theta}$
            \STATE Obtain candidate samples set $\mathcal{X}_d=\{\mathbf{x}_{t=0}^T\}$
            \STATE // Transition point detection
            \STATE $t^* = \text{sigmoid}(\omega \cdot \exp(k \cdot p_s)) \cdot T$
            \STATE Obtain embedding of transition point $\mathbf{x}_{d,t^*}$
            \STATE $\mathbf{e}_d = \mathbf{x}_{d,t^*}$
            \STATE // Train Encoder model
            \STATE Get triples $(u,i,j)$ from existing method
            \STATE Optimization Encoder by minimizing $\mathcal{L}_{d-bpr}$
        \ENDFOR
    \ENDWHILE
\end{algorithmic}
\end{algorithm}

\section{Experiments}
\subsection{Experimental Settings}
\subsubsection{Datasets.}
We conduct experiments on four benchmark datasets: Beauty, Toys, Sport, and Yelp, which are commonly used in both collaborative filtering (CF) and sequential recommendation (SR) tasks. All datasets are derived from the Amazon Review dataset and the Yelp Review dataset. The detailed statistics of the four datasets are summarized in Table~\ref{tab:dataset_stats}. 
For CF task, we follow standard preprocessing procedures \cite{lai2024ahns} by converting user behavior logs into user-item interaction matrices. Each user’s interactions are randomly split into 80\% for training and 20\% for testing. 
For SR task, user interactions are first sorted chronologically into sequences, and we filter out users with fewer than five interactions to ensure modeling stability.

\begingroup
\renewcommand{\arraystretch}{1.05}
\begin{table}[t]
\centering
\begin{tabularx}{\columnwidth}{cccccc}
\hline
Dataset & \#Users & \#Items & \begin{tabular}[c]{@{}c@{}}\#Inter-\\ actions\end{tabular} & \begin{tabular}[c]{@{}c@{}}Avg.\\ Length\end{tabular} & Density \\ \hline
Beauty  & 22,363  & 12,101  & 0.2m           & 8.9       & 0.05\%  \\
Toys    & 19,412  & 11,924  & 0.17m          & 8.6       & 0.07\%  \\
Sport   & 35,598  & 18,357  & 0.3m           & 8.3       & 0.05\%  \\
Yelp    & 30,431  & 20,033  & 0.3m           & 8.3       & 0.05\%  \\
\hline
\end{tabularx}
\caption{Statistics of the datasets.}
\label{tab:dataset_stats}
\end{table}
\endgroup

\begingroup
\setlength{\tabcolsep}{4pt}
\renewcommand{\arraystretch}{1.1}
\begin{table*}[!t]
    \centering
    \begin{tabular}{cc|ccc|ccc|ccc|ccc}
        \hline
        \multirow{2}{*}{Task} &\multirow{2}{*}{Methods} & \multicolumn{3}{c|}{Amazon-Beauty} & \multicolumn{3}{c|}{Amzon-Toys} & \multicolumn{3}{c|}{Amazon-Sport} & \multicolumn{3}{c}{Yelp} \\ \cline{3-14} 
        &  &R@10 &R@20 & N@10 & R@10 &R@20 & N@10 & R@10 &R@20 & N@10 & R@10 &R@20 & N@10 \\ 
        \hline
        \multirow{12}{*}{CF} &NGCF  &6.37  &9.50  &4.11  &5.94  &8.50  &3.92  &3.92  &6.21  &2.50  &4.27  &6.92  &2.73 \\ 
        &\textbf{+ADAR}  &\textbf{8.13}  &\textbf{11.76}  &\textbf{5.21}  &\textbf{7.82}  &\textbf{11.39}  &\textbf{5.08}  &\textbf{5.29}  &\textbf{8.11}  &\textbf{3.35}  &\textbf{5.06}  &\textbf{8.19}  &\textbf{3.24} \\ 
        &\texttt{RelImp}  &27.6\%  &23.8\%  &26.8\%  &31.6\%  &34.0\%  &29.6\%  &34.9\%  &30.6\%  &34.0\%  &18.5\%  &18.3\%  &18.7\% \\
        &LightGCN  &8.69  &12.40  &5.77  &8.22  &11.50  &5.55  &5.96  &8.62  &3.82  &5.27  &8.33  &3.39 \\ 
        &\textbf{+ADAR}  &\textbf{8.99}  &\textbf{12.88}  &\textbf{6.05}  &\textbf{8.46}  &\textbf{11.93}  &\textbf{5.84}  &\textbf{6.32}    &\textbf{9.08}  &\textbf{4.11}  &\textbf{5.55}  &\textbf{8.81}  &\textbf{3.58} \\ 
        &\texttt{RelImp}  &3.45\%  &3.55\%  &4.85\%  &2.92\%  &3.74\%  &5.23\%  &6.04\%  &5.34\%  &7.59\%  &5.31\%  &5.76\%  &5.60\% \\
        &MixGCF &9.82  &13.67  &6.63  &9.05  &12.82  &6.19  &6.25  &8.98  &4.21  &5.80  &9.01  &3.72 \\ 
        &\textbf{+ADAR}  &\textbf{10.01}  &\textbf{14.10}  &\textbf{6.90}  &\textbf{9.25}  &\textbf{13.09}  &\textbf{6.39}  &\textbf{6.31}  &\textbf{9.13}  &\textbf{4.32}  &\textbf{5.95}  &\textbf{9.31}  &\textbf{3.86} \\ 
        &\texttt{RelImp}  &1.93\%  &3.15\%  &4.07\%  &2.21\%  &2.11\%  &3.23\%  &0.96\%  &1.67\%  &2.61\%  &2.59\%  &3.33\%  &3.76\% \\
        &AHNS  &9.66  &13.50  &6.51  &8.89  &12.55  &6.06  &6.69  &9.70  & 4.39  &5.53  &8.88  &3.60 \\ 
        &\textbf{+ADAR}  &\textbf{10.05}  &\textbf{14.04}  &\textbf{6.89}  &\textbf{9.15}  &\textbf{12.91}  &\textbf{6.21}  &\textbf{6.84}  &\textbf{9.85}  &\textbf{4.52}  &\textbf{5.82}  &\textbf{9.19}  &\textbf{3.80} \\ 
        &\texttt{RelImp}  &4.04\%  &4.00\%  &5.84\%  &2.92\%  &2.87\%  &2.48\%  &2.24\%  &1.55\%  &2.96\%  &5.24\%  &3.49\%  &5.56\% \\ 
        \hline
        \hline
        \multirow{9}{*}{SR} &GRU4Rec  &2.79  &4.87  &1.37  &2.24  &3.71  &1.08  &1.81  &3.17  &0.90  &2.30  &3.88  &1.14 \\ 
        &\textbf{+ADAR}  &\textbf{4.34}  &\textbf{6.96}  &\textbf{2.17}  &\textbf{3.78}  &\textbf{5.84}  &\textbf{1.89}  &\textbf{2.55}  &\textbf{4.09}  &\textbf{1.39}  &\textbf{3.21}  &\textbf{5.17}  &\textbf{1.59} \\
        &\texttt{RelImp}  &55.6\%  &42.9\%  &58.4\%  &68.8\%  &57.4\%  &73.4\%  &40.9\%  &29.0\%  &54.4\%  &39.6\%  &33.2\%  &39.5\% \\      
        &SASRec  &5.99  &8.82  &3.20  &7.27  &9.95  &4.13  &2.97  &4.72  &1.58  &2.99  &4.99  &1.50 \\ 
        &\textbf{+ADAR}  &\textbf{6.32}  &\textbf{9.69}  &\textbf{3.27}  &\textbf{7.72}  &\textbf{10.52}  &\textbf{4.50}  &\textbf{3.62}  &\textbf{5.71}  &\textbf{1.89}  &\textbf{3.30}  &\textbf{5.58}  &\textbf{1.66}\\ 
        &\texttt{RelImp}  &5.51\%  &9.87\%  &2.19\%  &6.19\%  &5.73\%  &8.96\%  &21.9\%  &21.0\%  &19.6\%  &10.4\%  &11.8\%  &10.7\% \\
        &CL4SRec  &7.24  &10.26  &4.20  &7.80  &10.59  &4.54  &4.33  &6.19  &2.37  &3.27  &5.58  &1.60 \\ 
        &\textbf{+ADAR}  &\textbf{7.44}  &\textbf{10.65}  &\textbf{4.31}  &\textbf{8.12}  &\textbf{11.13}  &\textbf{4.73}  &\textbf{4.46}  &\textbf{6.35}  &\textbf{2.50}  &\textbf{3.69}  &\textbf{5.86}  &\textbf{1.82} \\ 
        &\texttt{RelImp}  &2.76\%  &3.80\%  &2.62\%  &4.10\%  &5.10\%  &4.19\%  &3.00\%  &2.58\%  &5.49\%  &12.8\%  &5.02\%  &13.8\% \\ 
        \hline
    \end{tabular}
    \caption{Performance on four datasets. 
    \texttt{RelImp}: Relative Improvement.
    We have conducted the paired t-test to verify that the difference between each base model and our proposed method is statistically significant for $p<0.05$.}
    \label{tab:main_results}
\end{table*}
\endgroup

\subsubsection{Evaluation Metrics.}
To assess recommendation performance, we adopt two widely used ranking metrics: Recall@{10, 20} and NDCG@{10, 20}. Recall measures the proportion of relevant items successfully retrieved in the top-$k$ recommendations, while NDCG accounts for the position of the hit items, giving higher weight to correctly ranked items appearing earlier in the list.

\subsubsection{Integrated Methods.}
To evaluate the effectiveness of our proposed ADAR, we conduct experiments under two major recommendation paradigms: collaborative filtering (CF) and sequential recommendation (SR).

For CF, we adopt the following models: NGCF \cite{wang2019ngcf}, LightGCN \cite{he2020lightgcn}, MixGCF \cite{huang2021mixgcf} and AHNS \cite{lai2024ahns}. 
For SR, we consider three representative sequential models with datasets augmented by DR4SR \cite{yin2024dr4sr}: 
GRU4Rec \cite{hidasi2015gru4rec}, SASRec \cite{kang2018sasrec} and CL4SRec \cite{xie2022cl4srec}.
See Appendix for a more detailed description of the baselines.

For each model, we incorporate our ADAR without modifying their core architectures, ensuring a fair and direct comparison with the original baselines.

\subsubsection{Implementation Details.}
The embedding dimension is fixed to 64, and the embedding parameters are initialized with the Xavier initialization \cite{glorot2010xavier}. 
We optimize all parameters with Adam optimizer \cite{kingma2014adam} and use the default learning rate of 0.001. For CF, we use the default mini-batch size of 2,048. For SR, we use the default mini-batch size of 256, with the maximum sequence length $N$ set to 50. 
For ADAR, we fix $\omega=1$ and $k=1$, and search $\lambda$ in the range $[0,1]$.
We set the maximum diffusion step $T$ to 50 for CF and 20 for SR.

\subsection{Result Analysis}
In this section, we compare the performance of each target model with our ADAR to verify the efficacy of the proposed module. 
The evaluation results for recommendations are summarized in Table~\ref{tab:main_results} (more details in Appendix), and we can draw the following conclusions:
\begin{itemize}
    \item \textbf{Effectiveness.} ADAR consistently improves performance across multiple datasets and baseline models. For example, AHNS enhanced by ADAR leads to a relative improvement of 3.13\% in Recall@10 and 4.58\% in NDCG@10 on Amazon-Beauty. Similar trends are observed across remaining methods, confirming its effectiveness in improving ranking quality.
    \item \textbf{Generality.} ADAR exhibits strong compatibility with diverse model architectures, including CF and SR models. ADAR can be compatible with models such as GRU4Rec (RNN-based), SASRec (attention-based), and CL4SRec (contrastive-based), achieving performance improvements without modifying the underlying model architecture.
    \item \textbf{Complementarity.} Beyond general applicability, ADAR also complements strong existing training objectives. For instance, ADAR improves CL4SRec’s NDCG@10 by 13.8\% on Yelp, indicating that it introduces additional beneficial training signals beyond those captured by contrastive learning alone.
\end{itemize}

\begingroup
\renewcommand{\arraystretch}{1.02}
\begin{table}[]
\begin{tabularx}{\linewidth}{l l *{4}{>{\centering\arraybackslash}X}}
\hline
Method                  & Version       & R@10 & R@20 & N@10 & N@20 \\ \hline
\multirow{4}{*}{NGCF}   & base          &6.37      &9.50      &4.11      &5.10      \\
                        & random        &7.81      &11.59      &5.13      &6.29      \\
                        & fixed         &7.79      &11.50      &5.06      &6.22      \\
                        & mixed         &7.91      &11.63      &5.11      &6.30      \\
                        & \textbf{ADAR}  &\textbf{8.13} &\textbf{11.76} &\textbf{5.21} &\textbf{6.34} \\ 
                        \hline
\multirow{4}{*}{GRU4Rec} & base         &2.79      &4.87      &1.37      &1.90      \\
                        & random        &4.00      &6.31      &1.99      &2.45      \\
                        & fixed         &3.37      &5.59      &1.66      &2.22      \\
                        & mixed         &3.51      &5.76      &1.77      &2.33      \\
                        & \textbf{ADAR}  &\textbf{4.34} &\textbf{6.96} &\textbf{2.17} &\textbf{2.82} \\ \hline
\end{tabularx}
\caption{Performance comparison of different $t^*$ selection strategies on Amazon-beauty.}
\label{tab:variant}
\end{table}
\endgroup

\subsection{Variant Study}
In this section, we investigate the impact of different strategies for selecting $t^*$. 
Recall that $t^*$ denotes the diffusion step of the final augmented negative, at which the generative sample translate from positive to negative. 
We compare five variants: 
(1) base, where no augmentation is applied; 
(2) random, where $t^*$ is randomly sampled; 
(3) fixed, where $t^*$ is fixed at a predefined midpoint ( we set it as $T$/2); 
(4) mixed, where model mixes representations from multiple fixed $t$ (e.g., $T/2$, $T/4$, $T/8$ and $T/10$), motivated by DMNS \cite{nguyen2024dmns} on link prediction;
(5) ADAR, the proposed method, which adaptively identifies $t^*$ based on theoretical criteria.

The experiments are conducted on Amazon-Beauty, evaluating two representative recommendation models: NGCF and GRU4Rec. 
We report the results in Table~\ref{tab:variant}, and make the following observations:
First, ADAR consistently outperforms all variants on both models, confirming that adaptive selection of $t^*$ yields more informative negatives than other variants.
Second, while effective in link prediction, the mixed variant is less suited to recommendation, performing worse than ADAR and even worse than the random variant.
This indicates that adaptive, theory-driven transition point is crucial for choosing high-quality negatives.

\begin{figure}
    \centering
    \begin{subfigure}[b]{0.23\textwidth}
        \centering
        \includegraphics[width=\textwidth]{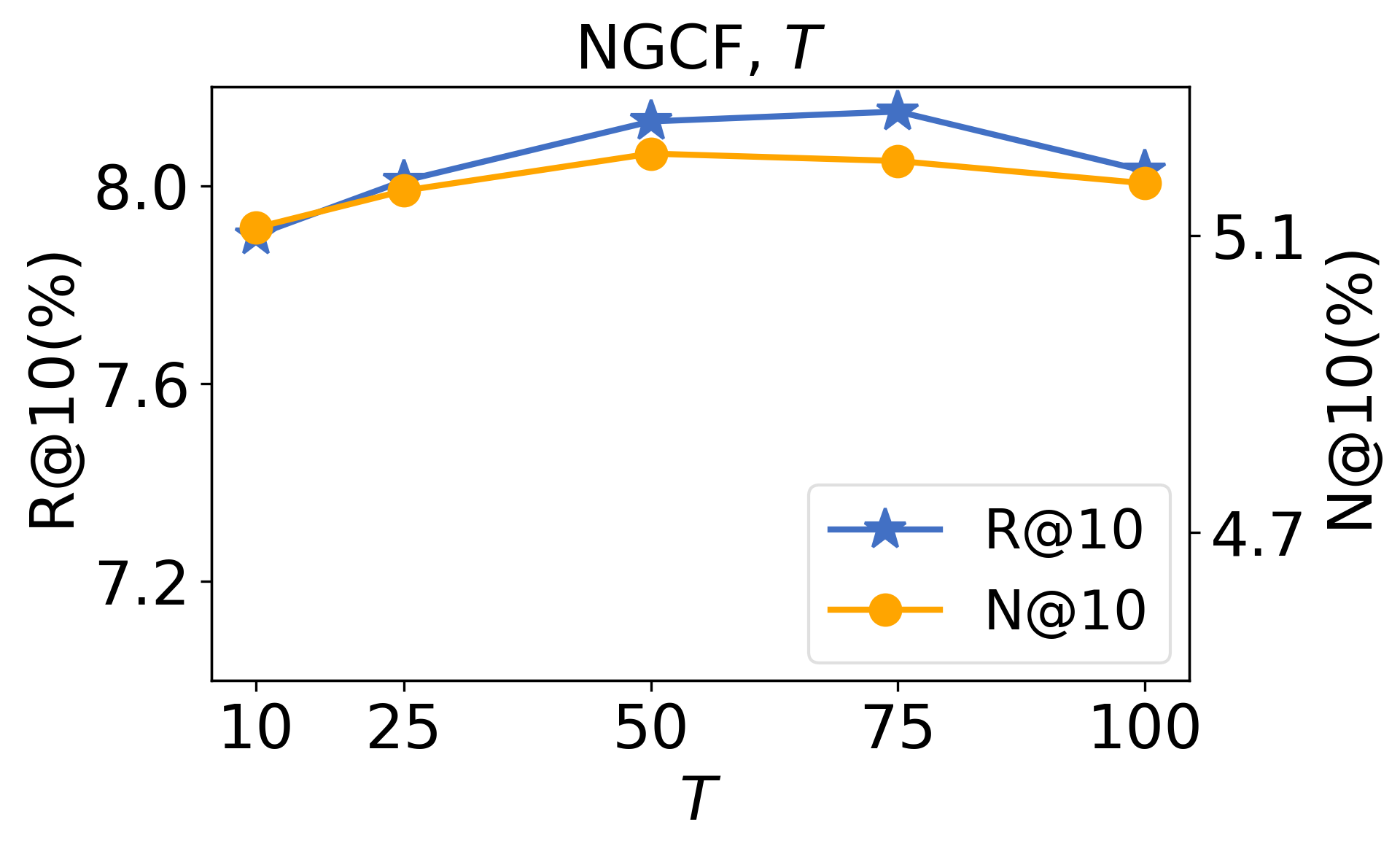}
        \caption{NGCF, $T$}
        \label{fig:t_ngcf}
    \end{subfigure}
    \hspace{0.1em}
    \begin{subfigure}[b]{0.23\textwidth}
        \centering
        \includegraphics[width=\textwidth]{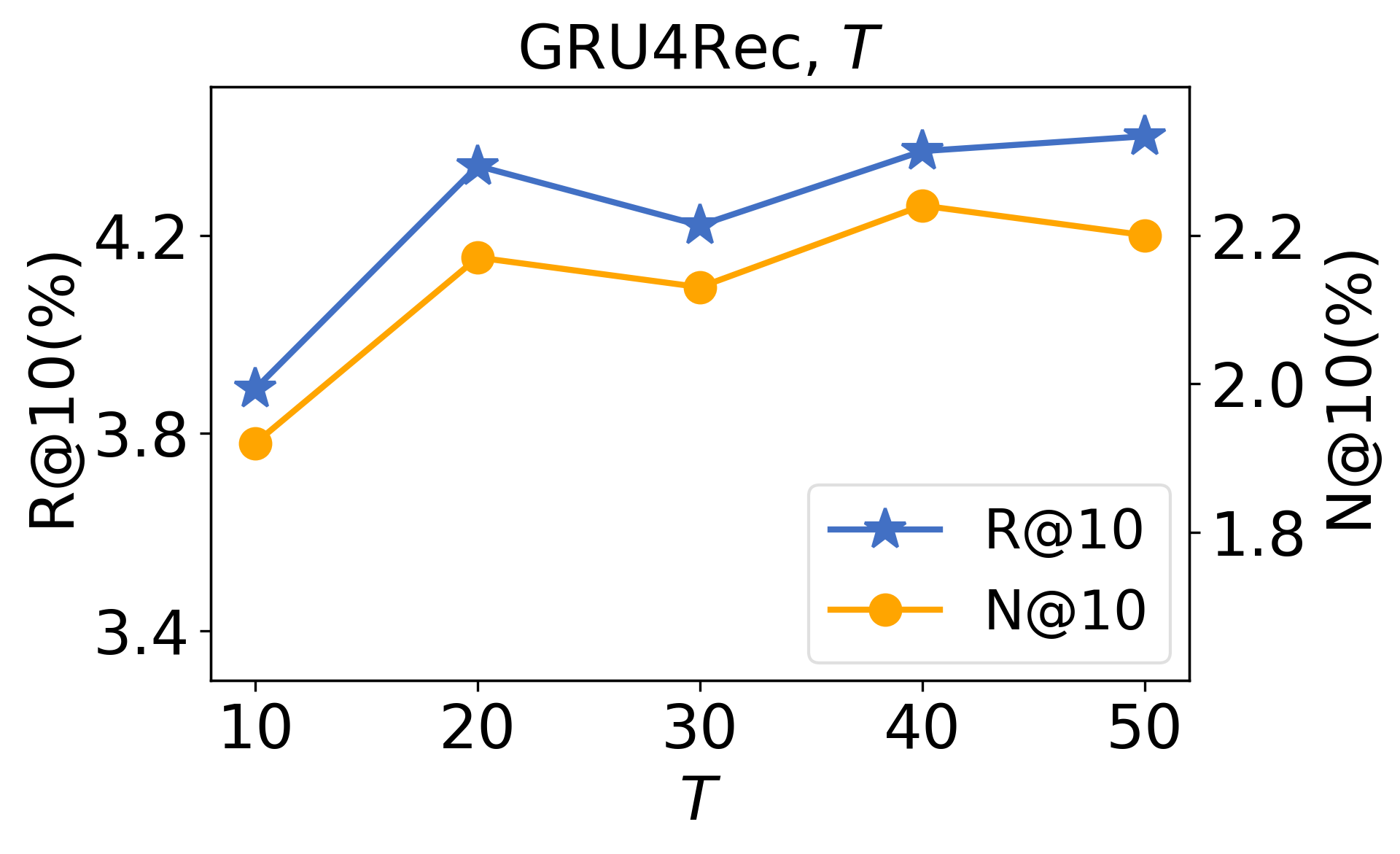}
        \caption{GRU4Rec, $T$}
        \label{fig:t_gru}
    \end{subfigure}

    \vspace{0.5em}
    
    \begin{subfigure}[b]{0.23\textwidth}
        \centering
        \includegraphics[width=\textwidth]{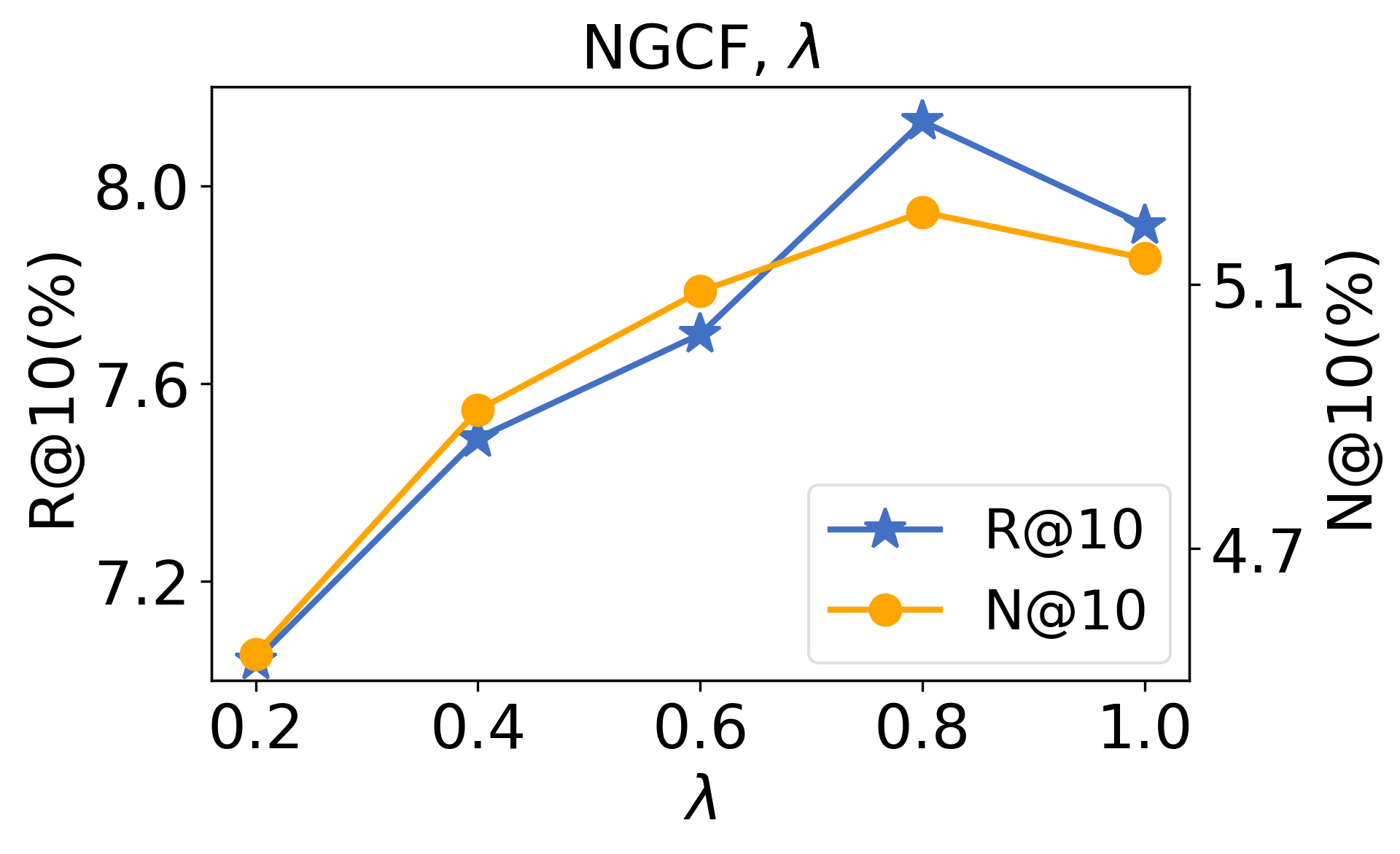}
        \caption{NGCF, $\lambda$}
        \label{fig:lamda_ngcf}
    \end{subfigure}
    \hspace{0.1em}
    \begin{subfigure}[b]{0.23\textwidth}
        \centering
        \includegraphics[width=\textwidth]{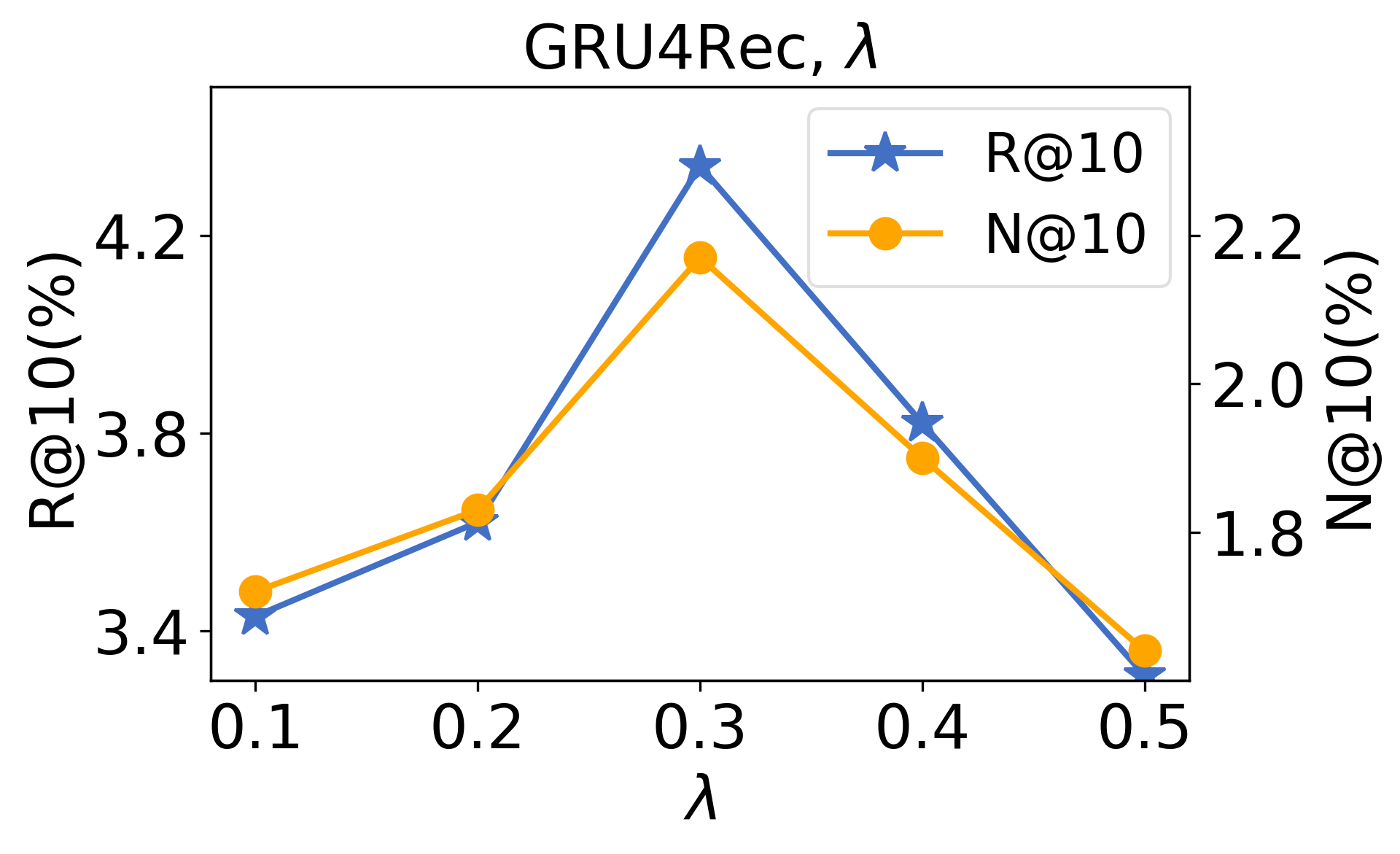}
        \caption{GRU4Rec, $\lambda$}
        \label{fig:lamda_gru}
    \end{subfigure}
    
    \caption{Performance comparison of ADAR \textit{w.r.t.} three different hyperparameter.}
    \label{fig:parameter}
\end{figure}

\subsection{Hyperparameter}
To further investigate robustness and stability of ADAR, we conduct a hyperparameter analysis on the diffusion step $T$ and weighting factor $\lambda$. The analysis of the hyperparameters $\omega$ and $k$ can be found in Appendix.
In this analysis, we fix the other hyperparameters to their default values to isolate the effect of each individual variable. 
The experiments are conducted on Amazon-beauty using two representative backbone models: NGCF for collaborative filtering and GRU4Rec for sequential recommendation. 
As shown in Figure~\ref{fig:parameter}(a) and \ref{fig:parameter}(b), the performance initially improves as $T$ increases, but eventually saturates, suggesting that only a moderate number of diffusion steps is sufficient to capture the semantic transition.
In contrast, Figure~\ref{fig:parameter}(c) and \ref{fig:parameter}(d) show that $\lambda$ exhibits a clear trade-off behavior, where moderate values consistently lead to better performance. 
This suggests that incorporating diffusion signal moderately can effectively guide negative sample selection.

\subsection{Visualization}
In this section, we delve into the visualization analysis to further examine the effectiveness of ADAR in capturing meaningful user preferences. Specifically, we randomly sample 2,000 users from Amazon-Beauty and project their learned user embeddings into a two-dimensional space. 
This is achieved by normalizing each representation onto the unit hypersphere and applying t-SNE \cite{maaten2008tsne} for dimensionality reduction. 
As shown in Figure~\ref{fig:vis}, user representations learned by baseline models (NGCF and SASRec) tend to exhibit uneven and clustered patterns. In contrast, their ADAR-augmented counterparts (NGCF with ADAR and SASRec with ADAR) yield embeddings that are more uniformly distributed across the space. This suggests that incorporating diffusion-based augmentation encourages a more informative and disentangled user representation space, potentially leading to improved generalization.

\begin{figure}
    \centering
    \begin{subfigure}[b]{0.23\textwidth}
        \centering
        \includegraphics[width=\textwidth]{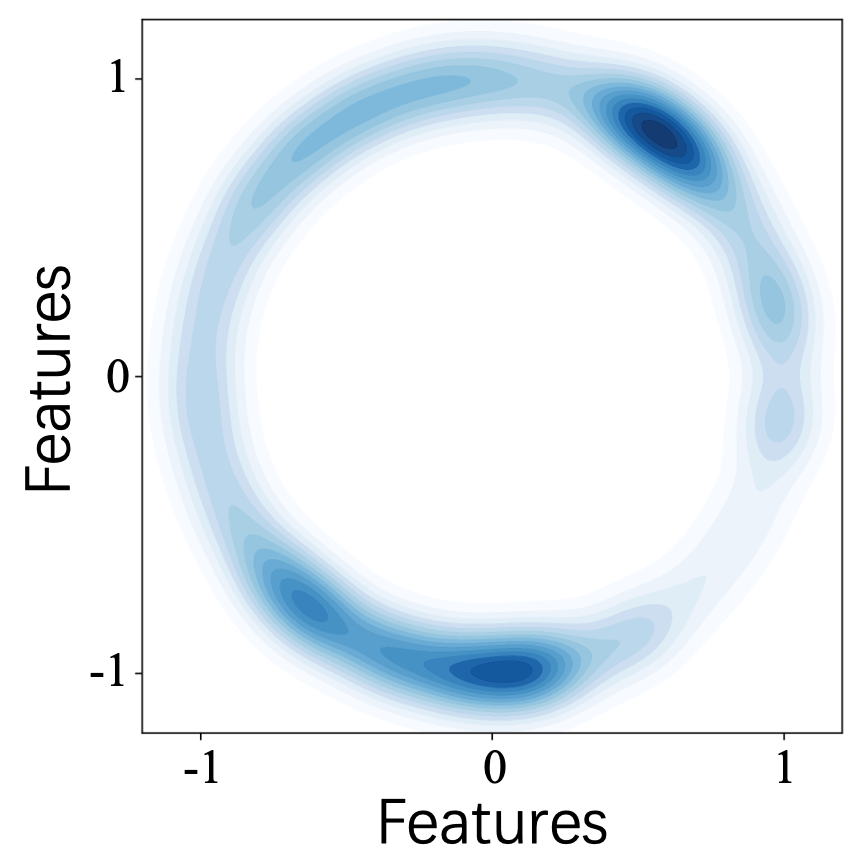}
        \caption{NGCF}
        \label{fig:ngcf}
    \end{subfigure}
    \hspace{0.1em}
    \begin{subfigure}[b]{0.23\textwidth}
        \centering
        \includegraphics[width=\textwidth]{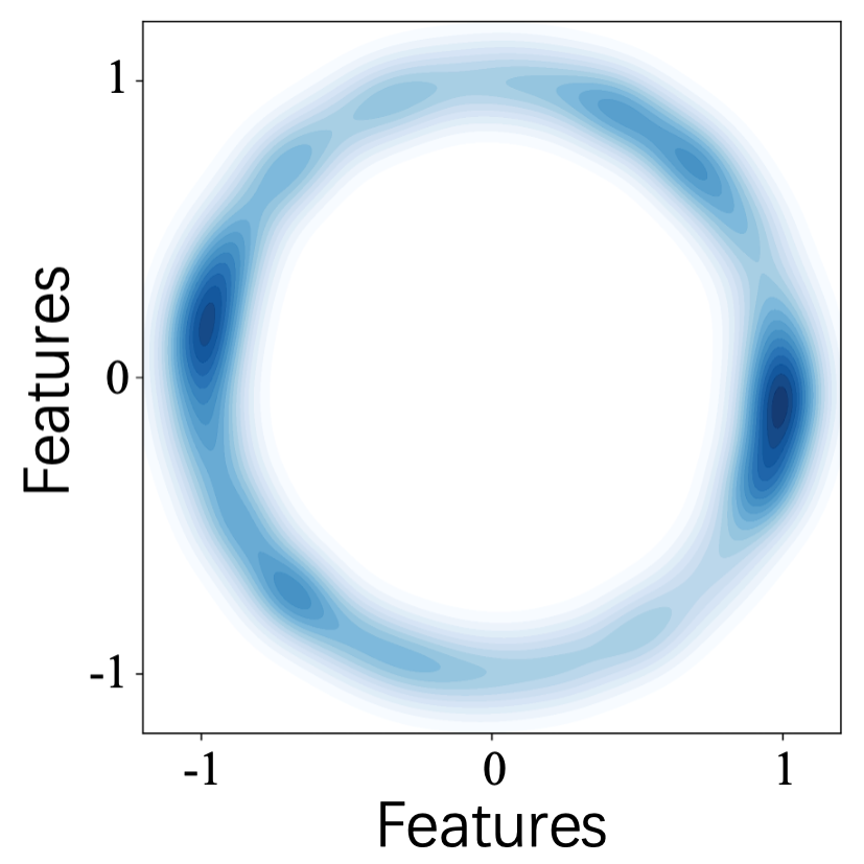}
        \caption{NGCF with ADAR}
        \label{fig:dra-ngcf}
    \end{subfigure}

    \vspace{0.5em}
    
    \begin{subfigure}[b]{0.23\textwidth}
        \centering
        \includegraphics[width=\textwidth]{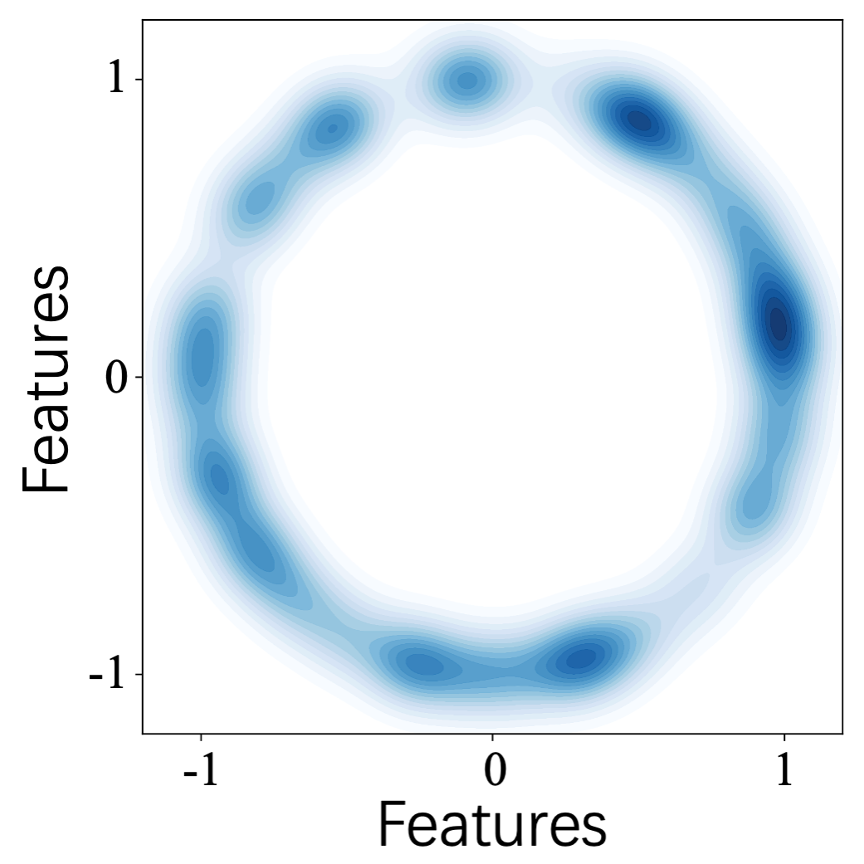}
        \caption{SASRec}
        \label{fig:sasrec}
    \end{subfigure}
    \hspace{0.1em}
    \begin{subfigure}[b]{0.23\textwidth}
        \centering
        \includegraphics[width=\textwidth]{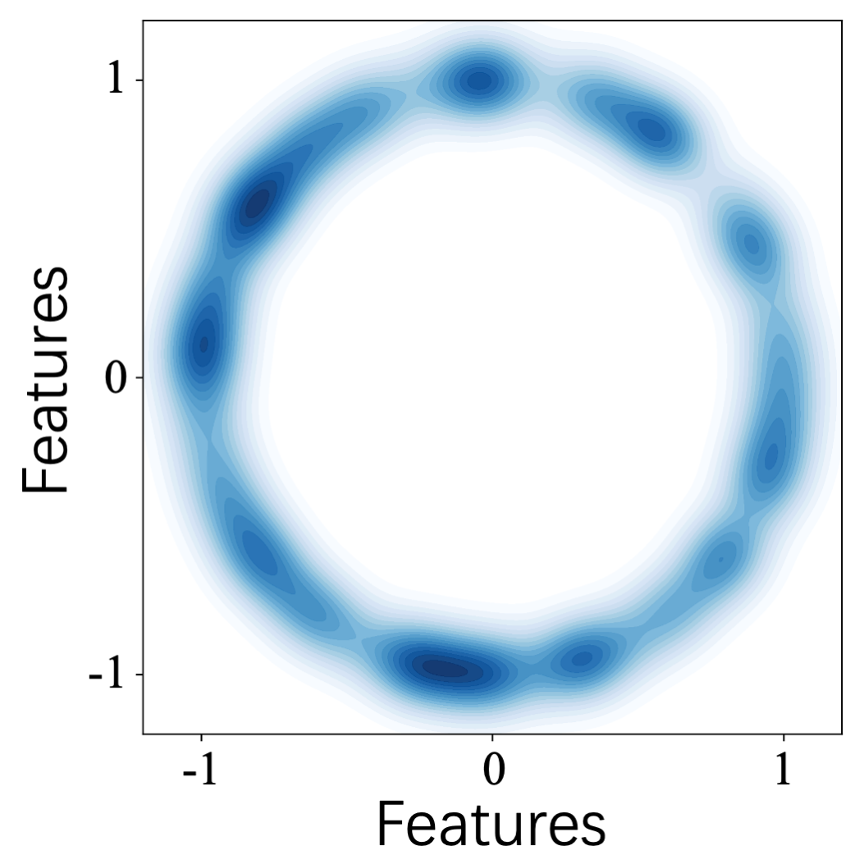}
        \caption{SASRec with ADAR}
        \label{fig:dra-sasrec}
    \end{subfigure}
    
    \caption{Distribution of user representations learned from Amazon-beauty dataset.}
    \label{fig:vis}
\end{figure}

\section{Conclusion}
To solve the challenge of unreliable negative signals and limited control over negative sampling in implicit feedback recommendation, we propose ADAR, a diffusion-based augmentation module.
ADAR leverages the progressive corruption process of diffusion models to generate high-quality negative samples. 
By theoretically identifying the transition point from positive to negative, ADAR adaptively selects optimal sampling steps with informative and challenging negatives.
Extensive experiments on various backbone models and recommendation tasks demonstrate that ADAR is both effective and broadly applicable, offering a robust solution for enhancing recommendation performance.
Our code is publically available at https://github.com/LN-Nlaine/ADAR.

\section*{Acknowledgments}
This work was supported by Natural Science Foundation of Heilongjiang Province of China (No. LH2024F023).

\bibliography{aaai2026}

\end{document}